\documentclass[conference]{IEEEtran}
\usepackage[boxruled]{algorithm2e}
\usepackage{cite}
\usepackage{graphicx}
\usepackage{epstopdf}
\usepackage{psfrag}
\usepackage{subfigure}
\usepackage{url}
\usepackage{amsmath}
\usepackage{array}
\usepackage{amssymb}
\usepackage{amsfonts}
\usepackage{epstopdf}

\newtheorem{lemma}{Lemma}

\newtheorem{corollary}{Corollary}

\newtheorem{theorem}{Theorem}

\title{Capacity Region of $K$-User Discrete Memoryless Interference Channels with a Mixed Strong-Very Strong Interference}
\begin{document}

\author{
\authorblockN{G. Abhinav}
\authorblockA{Dept. of ECE, Indian Institute of Science \\
Bangalore 560012, India\\
Email: abhig\_88@ece.iisc.ernet.in
}
\and
\authorblockN{B. Sundar Rajan}
\authorblockA{Dept. of ECE, Indian Institute of Science, \\Bangalore 560012, India\\
Email: bsrajan@ece.iisc.ernet.in
}
}

\maketitle
\thispagestyle{empty}	
%%%%%%%%

\begin{abstract}
The capacity region of the 3-user Gaussian Interference Channel (GIC) with mixed strong-very strong interference was established in \cite{ChS}. The mixed strong-very strong interference conditions considered in \cite{ChS} correspond to the case where, at each receiver, one of the interfering signals is strong and the other is very strong. In this paper, we derive the capacity region of $K$-user $(K\geq 3)$ Discrete Memoryless Interference Channels (DMICs) with a mixed strong-very strong interference. This  corresponds to the case where, at each receiver one of the interfering signals is strong and the other $(K-2)$ interfering signals are very strong. This includes, as a special case, the 3-user DMIC with mixed strong-very strong interference. The proof is specialized to the 3-user GIC case and hence an alternative simpler derivation for the capacity region of the 3-user GIC with mixed strong-very strong interference is provided.

\end{abstract}	
\section{INTRODUCTION}
The capacity of a general K-user interference channel has been open for decades. The capacity region for the 2-user Gaussian Interference Channel (GIC) with strong interference was established in \cite{HS} and the capacity region for the 2-user discrete memoryless interference channel (DMIC) with strong interference was derived in \cite{CoE}.  The sum-capacity of the 2-user GIC was obtained for a noisy interference regime \cite{SKC1,AnV,MoK} where, treating interference as noise at each receiver achieves the sum-capacity. In general, for the 2-user GIC, the capacity region is known within a gap of one bit \cite{ETW}.

Recently, there has also been some progress in characterizing the capacity region of the interference channel for more than 2-users. In \cite{SJVJ}, lattice codes were used to achieve the capacity region of the K-user symmetric Gaussian very strong interference channel. In \cite{SKC2}, a noisy interference regime of the K-user IC was derived as an extension of the 2-user result. In \cite{JoV}, the sum-capacity of K-user degraded GIC has been derived and the scheme that achieves this sum-capacity is shown to be successive interference cancellation.

This paper is based on the work in \cite{ChS} where, the capacity region of the 3-user GIC with mixed strong-very strong interference has been derived. This was defined as the condition where, at each receiver, one of the interfering signals is strong and the other interfering signal is very strong. 

The contributions and organization of this paper are as follows:
\begin{itemize}
\item The capacity region of the 3-user DMIC with mixed strong-very strong interference is established (Theorem \ref{thm1} and Theorem \ref{thm2} in Section \ref{sec3}).
\item The capacity region of the 3-user GIC with mixed strong-very strong interference is established (Corollary \ref{cor1} in Section \ref{sec4}). Our proof for this is much simpler than the proof in \cite{ChS}.
\item The capacity region for the 3-user DMIC with mixed strong-very strong interference is generalized to the $K$-user scenario $(K\geq 3)$ where, at each receiver, one of the interfering signals is strong and the other $(K-2)$ interfering signals are very strong\footnote{For $K$=$2$, the capacity region given in Theorem \ref{thm3} reduces to the capacity region of GIC with strong interference \cite{HS}.} (Theorem \ref{thm3} in Section \ref{sec5}).
\end{itemize}

In the next section, we present the channel model.

\textit{\textbf{Notations:}} Realization of an alphabet $X$ is denoted as $x$. The probability distribution on the alphabet $X$ is denoted by $p_X(x)$. ${\cal CN}(0,\sigma^2)$ represents circularly symmetric complex Gaussian noise with mean $0$ and variance $\sigma^2$. For a random-variable $Q$, $|Q|$ denotes the cardinality of the support-set from which $Q$ can take values. $\mathbb{C}$ denotes the set of complex numbers.

%%%%%%%%%%%%%%%%%%%%%%%%%%%%%%%%%%%%%%%%%%%%%%%%%
\section{CHANNEL MODEL} 
\label{sec2}
The 3-user DMIC model considered in this paper is shown in Fig. \ref{fig:DMIC}.
%%%%%%%%%%%%%%%%%%%%%%%%%%%%%%%%%%%%%%%%%%%%%%%%%%%%%%%%%%%%%%%%%
\begin{figure}[htbp]
\vspace{-1cm}
\centering
\includegraphics[totalheight=4in,width=3.0in]{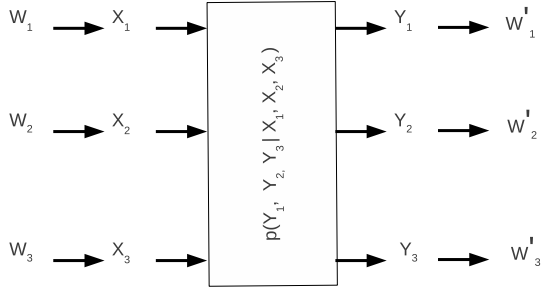}
\vspace{-4cm}
\caption{3-User DMIC Model}	
\label{fig:DMIC}	
\end{figure}
%%%%%%%%%%%%%%%%%%%%%%%%%%%%%%%%%%%%%%%%%%%%%%%%%%%%%%%%%%%%%%%%%
The channel input from User-$i$ is denoted by $X_i$ and $Y_i$ is the channel output at Receiver-$i$, $i$ $\in$ $\{1, 2, 3\}$, and all of them take values from finite alphabets. User-$i$ intends to communicate with Receiver-$i$ at rate $R_i$, $i$ $\in$ $\{1,2,3\}$, through a memoryless channel with transition probability $p(Y_1 ,Y_2 ,Y_3 |X_1 ,X_2 ,X_3)$. User-$i$, $i$ $\in$ $\{1, 2, 3\}$, encodes its independent message $W_i$ into a codeword of length $n$, i.e., $X_i^n$. We assume that the message $W_i$, $i$ $\in$ $\{1,2,3\}$, is uniformly distributed over a set of cardinality $2^{nR_i}$.

Receiver-$i$, $i$ $\in$ $\{1,2,3\}$, assigns an estimate $W_i^\prime$ to each received sequence $Y_i^n$. The average probability of error is defined by $P_e^{(n)} = P\{(W_1^{\prime}, W_2^{\prime}, W_3^{\prime}) \neq (W_1, W_2, W_3)\}$. A rate triplet $(R_1, R_2, R_3)$ is said to be achievable for the DMIC if there exists a sequence of $(2^{nR_1}, 2^{nR_2}, 2^{nR_3}, n)$ codes with $P_e^{(n)} \rightarrow 0$. The capacity region of the DMIC is the closure of the set of all achievable rate triplets $(R_1, R_2, R_3)$.

The codes, achievable rates and average probability of error can be similarly defined for the $K$-user case.

\section{CAPACITY REGION OF $3$-USER DMIC WITH MIXED STRONG-VERY STRONG INTERFERENCE} 
\label{sec3}
In this section, we derive the capacity region of the 3-user DMIC with mixed strong-very strong interference, i.e., each receiver is constrained to receive a strong interference and a very strong interference. For the sake of clarity, we shall consider $2$ cases of the 3-user DMIC with mixed strong-very strong interference separately, as stated below (these are the only cases possible for the 3-user DMIC).
\begin{enumerate}
\item In the first case, each user causes a strong interference at one of the unintended receivers and a very strong interference at the other unintended receiver while satisfying the constraint that each receiver has to receive a strong interference and a very strong interference.
\item In the second case, one of the three users produces strong interference at both the unintended receivers, the second user causes very strong interference at both the unintended receivers and the remaining user produces a strong interference at one of the unintended receivers and a very strong interference at the other unintended receiver while satisfying the constraint that each receiver has to receive a strong interference and a very strong interference.
\end{enumerate}
Precise definitions for very strong interference and strong interference for Case $1$ are given in (\ref{eqn1})-(\ref{eqn3}) and (\ref{eqn4})-(\ref{eqn6}) respectively and for Case $2$ in (\ref{eqn13})-(\ref{eqn15}) and (\ref{eqn16})-(\ref{eqn18}) respectively. 

The following lemma will be useful in proving the results.
%%%%%%%%%%%%%%%%%%%%%%%%%%%%%%%%%%%%%%%%%%%%%%%
\begin{lemma}[\cite{CoE}]
\label{lemma1}
If $I[X_1;Y_1|X_2,X_3] \leq I[X_1;Y_2|X_2,X_3]$ $\forall ~p_{X_1}(x_1)p_{X_2}(x_2)p_{X_3}(x_3)$, then $I[X_1^n;Y_1^n|X_2^n,X_3^n] \leq I[X_1^n;Y_2^n|X_2^n,X_3^n]~ \forall ~n \geq 1$.
\end{lemma}
\begin{proof}
This can be proved in the same way as Lemma 1 in \cite{CoE}.
\end{proof}
%%%%%%%%%%%%%%%%%%%%%%%%%%%%%%%%%%%%%%%%%%%%%%%

\textbf{Case 1:} Without loss of generality, we assume that User-$1$ causes a very strong interference at Receiver-$2$ and strong interference at Receiver-$3$, User-$2$ causes a very strong interference at Receiver-$3$ and strong interference at Receiver-$1$, and User-$3$ causes a very strong interference at Receiver-$1$ and strong interference at Receiver-$2$, i.e., the conditions
%%%%%%%%%%%%%%%%%%%%%%%%%%%%%%%%%%%%%%%%%%%%%%%%%%%%%%%%%%

{\small
\begin{align}
\label{eqn1}
&I[X_1;Y_2] \geq I[X_1;Y_1|X_2,X_3]\\
\label{eqn2}
&I[X_2;Y_3] \geq I[X_2;Y_2|X_3,X_1]\\
\label{eqn3}
&I[X_3;Y_1] \geq I[X_3;Y_3|X_1,X_2]\\
\label{eqn4}
&I[X_1;Y_1|X_2,X_3] \leq I[X_1;Y_3|X_2,X_3]\\
\label{eqn5}
&I[X_2;Y_2|X_3,X_1] \leq I[X_2;Y_1|X_3,X_1]\\
\label{eqn6}
&I[X_3;Y_3|X_1,X_2] \leq I[X_3;Y_2|X_1,X_2]
\end{align}}are satisfied $\forall$ $p_{X_1}(x_1)p_{X_2}(x_2)p_{X_3}(x_3)$. The equations (\ref{eqn1})-(\ref{eqn3}) represent the very strong interference conditions and (\ref{eqn4})-(\ref{eqn6}) represent the strong interference conditions. Now, we establish the capacity region of this channel.
%%%%%%%%%%%%%%%%%%%%%%%%%%%%%%%%%%%%%%%%%%%%%%%%%%%%%%%%%%%
\begin{theorem}
\label{thm1}
The capacity region of the 3-user DMIC with mixed strong-very strong interference, satisfying (\ref{eqn1})-(\ref{eqn6}) is given by
%%%%%%%%%%%

{\small
\begin{align}
\label{eqn7}
&R_1<I[X_{1};Y_{1}|X_{2},X_{3},Q]\\
\label{eqn8}
&R_2<I[X_{2};Y_{2}|X_{3},X_{1},Q]\\
\label{eqn9}
&R_3<I[X_{3};Y_{3}|X_{1},X_{2},Q]\\
\label{eqn10}
&R_1+R_2<I[X_{1},X_{2};Y_{1}|X_{3},Q]\\
\label{eqn11}
&R_2+R_3<I[X_{2},X_{3};Y_{2}|X_{1},Q]\\
\label{eqn12}
&R_3+R_1<I[X_{3},X_{1};Y_{3}|X_{2},Q]
\end{align}}for some joint distribution $p_Q(q)p_{X_1|Q}(x_1|q)p_{X_2|Q}(x_2|q)$ $p_{X_3|Q}(x_3|q)$, where, $Q$ is the time-sharing random-variable with $|Q| \leq 7$.
\end{theorem}
\begin{proof}
We shall first prove the converse. 

\textit{Converse:} Equations (\ref{eqn7})-(\ref{eqn9}) are the usual outer bounds. Now,
%%%%%%%%%%%

{\small
\begin{align}
\nonumber
n(R_1+R_2) &\stackrel{(a)}{<} I[X_1^n;Y_1^n]+I[X_2^n;Y_2^n]+n \epsilon_n\\
\nonumber
&<I[X_1^n;Y_1^n,X_3^n]+I[X_2^n;Y_2^n,X_1^n,X_3^n]+n \epsilon_n\\
\nonumber
&\stackrel{(b)}{=}I[X_1^n;Y_1^n|X_3^n]+I[X_2^n;Y_2^n|X_1^n,X_3^n]+n \epsilon_n\\
\nonumber
&\stackrel{(c)}{\leq}I[X_1^n;Y_1^n|X_3^n]+I[X_2^n;Y_1^n|X_1^n,X_3^n]+n \epsilon_n\\
\nonumber
&=I[X_1^n,X_2^n;Y_1^n|X_3^n]+n \epsilon_n\\
\nonumber
&\stackrel{(d)}{=} h(Y_1^n|X_3^n)-\sum_{i=1}^{n} h(Y_{1i}|X_{1i},X_{2i},X_{3i})+n \epsilon_n\\
\nonumber
&\stackrel{(e)}{<} \sum_{i=1}^{n}\left( h(Y_{1i}|X_{3i})- h(Y_{1i}|X_{1i},X_{2i},X_{3i}) \right)\\
\nonumber
&~~~~~~~~~~~~~~~~~~~~~~~~~~~~~~~~~~~~~~~~~~~~~~~~~+n \epsilon_n\\
\nonumber
&=\sum_{i=1}^{n}I[X_{1i},X_{2i};Y_{1i}|X_{3i}]+n \epsilon_n\\
\nonumber
&=n\frac{1}{n}\sum_{i=1}^{n}I[X_{1i},X_{2i};Y_{1i}|X_{3i}]+n \epsilon_n\\
\nonumber
\Rightarrow R_1+R_2&< I[X_1,X_2;Y_1|X_3,Q]+ \epsilon_n
\end{align}}where, $(a)$ follows from Fano's inequality, $(b)$ follows from the fact that $X_1^n$, $X_2^n$ and $X_3^n$ are independent, $(c)$ follows from (\ref{eqn5}) and Lemma \ref{lemma1}, $(d)$ follows from the memoryless property of the channel and $(e)$ follows from the fact that removing conditioning increases entropy. Finally, taking the limit as $n \rightarrow \infty$, $P_e^{(n)} \rightarrow 0$, we have

{\small
\begin{align}
\nonumber
R_1+R_2&< I[X_1,X_2;Y_1|X_3,Q].
\end{align}}
Similarly, using (\ref{eqn4}) and (\ref{eqn6}), we can easily derive the outer bounds in (\ref{eqn11}) and (\ref{eqn12}) respectively.

\textit{Achievability:} Fix $p_{Q}(q)p_{X_1|Q}(x_1|q)p_{X_2|Q}(x_2 |q)$ $p_{X_3|Q}(x_3 |q)$. At User-$i$, $i \in \{1,2,3\}$, generate $2^{nR_i}$ independent codewords $X_i^n (W_{ij})$, $j \in \{1,2,3,..2^{nR_i}\}$, of length $n$, generating each element i.i.d. $\sim \prod_{k=1}^{n} p_{X_i|Q} (x_{ik}|q)$, where, $W_{ij}$ indicates message-$j$ at User-$i$. First decode User-$1$'s message at Receiver-$2$, User-$2$'s message at Receiver-$3$ and User-$3$'s message at Receiver-$1$ by the usual weak-typical set decoding. They can be decoded with arbitrarily small probability of error if

{\small
\begin{align}
\label{redund1}
&R_1<I[X_1;Y_2|Q];~R_2<I[X_2;Y_3|Q];~R_3<I[X_3;Y_1|Q].
\end{align}}
Now, given that the messages of User-$1$, User-$2$ and User-$3$ are known at Receiver-$2$, Receiver-$3$ and Receiver-$1$ respectively, we perform the usual MAC-type decoding \cite{CoT-book} on messages of User-$1$ and User-$2$ at Receiver-$1$, User-$2$ and User-$3$ at Receiver-$2$ and User-$3$ and User-$1$ at Receiver-$3$. They can be decoded with arbitrarily small probability of error if (\ref{eqn7})-(\ref{eqn12}) and
%%%%%%%%%%%%%%%%%%%%%%%%%%%%%%%%%%%%%%%%%%%%%

{\small
\begin{align}
\nonumber
&R_1 < I[X_1;Y_3|X_2,X_3,Q]\\
\label{redund2}
&R_2 < I[X_2;Y_1|X_3,X_1,Q]\\
\nonumber
&R_3 < I[X_3;Y_2|X_1,X_2,Q]
\end{align}}are satisfied. The conditions in (\ref{redund1}) and (\ref{redund2}) are redundant because (\ref{eqn1})-(\ref{eqn6}) are satisfied $\forall ~p_{X_1}(x_1)p_{X_2}(x_2)p_{X_3}(x_3)$. The cardinality of $Q$ follows from direct application of Caratheodory Theorem \cite{CoT-book}.
\end{proof}

\textbf{Case 2:} Without loss of generality, we assume that User-$1$ causes strong interference at both Receiver-$2$ and Receiver-$3$, User-$3$ causes very strong interference at both Receiver-$1$ and Receiver-$2$ while User-$2$ causes strong interference at Receiver-$1$ and very strong interference at Receiver-$3$, i.e., the conditions
%%%%%%%%%%%%%%%%%%%%%%%%%%%%%%%%%%%%%%%%%%%%%%%%%%%%%%%%%%

{\small
\begin{align}
\label{eqn13}
&I[X_3;Y_2] \geq I[X_3;Y_3|X_1,X_2]\\
\label{eqn14}
&I[X_2;Y_3] \geq I[X_2;Y_2|X_3,X_1]\\
\label{eqn15}
&I[X_3;Y_1] \geq I[X_3;Y_3|X_1,X_2]\\
\label{eqn16}
&I[X_1;Y_1|X_2,X_3] \leq I[X_1;Y_3|X_2,X_3]\\
\label{eqn17}
&I[X_2;Y_2|X_3,X_1] \leq I[X_2;Y_1|X_3,X_1]\\
\label{eqn18}
&I[X_1;Y_1|X_2,X_3] \leq I[X_1;Y_2|X_2,X_3]
\end{align}}are satisfied $\forall$ $p_{X_1}(x_1)p_{X_2}(x_2)p_{X_3}(x_3)$. The equations (\ref{eqn13})-(\ref{eqn15}) represent the very strong interference conditions and (\ref{eqn16})-(\ref{eqn18}) represent the strong interference conditions. Now, we establish the capacity region of this channel.
%%%%%%%%%%%%%%%%%%%%%%%%%%%%%%%%%%%%%%%%%%%%%%%%%%%%%%%%%%%
\begin{theorem}
\label{thm2}
The capacity region of the 3-user DMIC with mixed strong-very strong interference, satisfying (\ref{eqn13})-(\ref{eqn18}) is given by
%%%%%%%%%%%

{\small
\begin{align}
\label{eqn19}
&R_1<I[X_{1};Y_{1}|X_{2},X_{3},Q]\\
\label{eqn20}
&R_2<I[X_{2};Y_{2}|X_{3},X_{1},Q]\\
\label{eqn21}
&R_3<I[X_{3};Y_{3}|X_{1},X_{2},Q]\\
\label{eqn22}
&R_1+R_2<\min\{I[X_{1},X_{2};Y_{1}|X_{3},Q], I[X_{1},X_{2};Y_{2}|X_{3},Q]\}\\
\label{eqn23}
&R_3+R_1<I[X_{3},X_{1};Y_{3}|X_{2},Q]
\end{align}}for some joint distribution $p_Q(q)p_{X_1|Q}(x_1|q)p_{X_2|Q}(x_2|q)$ $p_{X_3|Q}(x_3|q)$, where, $Q$ is the time-sharing random-variable with $|Q| \leq 7$.
\end{theorem}
\begin{proof}
We shall first prove the converse.

\textit{Converse:} Equations (\ref{eqn19})-(\ref{eqn21}) are the usual outer bounds. Note that the condition (\ref{eqn5}) of Case 1 holds here too (in (\ref{eqn17})). Hence, the bound 
\begin{equation}
 \nonumber
R_1+R_2<I[X_{1},X_{2};Y_{1}|X_{3},Q]
\end{equation}is valid here too. Now,
%%%%%%%%%%%

{\small
\begin{align}
\nonumber
n(R_1+R_2) &\stackrel{(a)}{<} I[X_1^n;Y_1^n]+I[X_2^n;Y_2^n]+n \epsilon_n\\
\nonumber
&<I[X_1^n;Y_1^n,X_2^n,X_3^n]+I[X_2^n;Y_2^n,X_3^n]+n \epsilon_n\\
\nonumber
&\stackrel{(b)}{=}I[X_1^n;Y_1^n|X_2^n,X_3^n]+I[X_2^n;Y_2^n|X_3^n]+n \epsilon_n\\
\nonumber
&\stackrel{(c)}{\leq}I[X_1^n;Y_2^n|X_2^n,X_3^n]+I[X_2^n;Y_2^n|X_3^n]+n \epsilon_n\\
\nonumber
&=I[X_1^n,X_2^n;Y_2^n|X_3^n]+n \epsilon_n\\
\nonumber
&\stackrel{(d)}{=} h(Y_2^n|X_3^n)-\sum_{i=1}^{n} h(Y_{2i}|X_{1i},X_{2i},X_{3i})+n \epsilon_n\\
\nonumber
&\stackrel{(e)}{<} \sum_{i=1}^{n}\left( h(Y_{2i}|X_{3i})- h(Y_{2i}|X_{1i},X_{2i},X_{3i}) \right)\\
\nonumber
&~~~~~~~~~~~~~~~~~~~~~~~~~~~~~~~~~~~~~~~~~~~~~~~~~+n \epsilon_n\\
\nonumber
&=\sum_{i=1}^{n}I[X_{1i},X_{2i};Y_{2i}|X_{3i}]+n \epsilon_n\\
\nonumber
&=n\frac{1}{n}\sum_{i=1}^{n}I[X_{1i},X_{2i};Y_{2i}|X_{3i}]+n \epsilon_n\\
\Rightarrow R_1+R_2&< I[X_1,X_2;Y_2|X_3,Q]+ \epsilon_n
\end{align}}where, $(a)$ follows from Fano's inequality, $(b)$ follows from the fact that $X_1^n$, $X_2^n$ and $X_3^n$ are independent, $(c)$ follows from (\ref{eqn18}) and Lemma \ref{lemma1}, $(d)$ follows from the memoryless property of the channel and $(e)$ follows from the fact that removing conditioning increases entropy. Finally, taking the limit as $n \rightarrow \infty$, $P_e^{(n)} \rightarrow 0$, we have 

{\small
\begin{align}
\nonumber
R_1+R_2&< I[X_1,X_2;Y_2|X_3,Q].
\end{align}}%
Similarly, using (\ref{eqn16}), we can easily derive the outer bound in (\ref{eqn23}).

\textit{Achievability:} Fix $p_{Q}(q)p_{X_1|Q}(x_1|q)p_{X_2|Q}(x_2 |q)$ $p_{X_3|Q}(x_3 |q)$. At User-$i$, $i \in \{1,2,3\}$, generate $2^{nR_i}$ independent codewords $X_i^n (W_{ij})$, $j \in \{1,2,3,..2^{nR_i}\}$, of length $n$, generating each element i.i.d. $\sim \prod_{k=1}^{n} p_{X_i|Q} (x_{ik}|q)$, where, $W_{ij}$ indicates message-$j$ at User-$i$. First decode User-$3$'s message at Receiver-$2$, User-$2$'s message at Receiver-$3$ and User-$3$'s message at Receiver-$1$ by the usual weak-typical set decoding. They can be decoded with arbitrarily small probability of error if

{\small
\begin{align}
\label{redund3}
&R_3<I[X_3;Y_2|Q];~R_2<I[X_2;Y_3|Q];~R_3<I[X_3;Y_1|Q].
\end{align}}
Now, given that the messages of User-$3$, User-$2$ and User-$3$ are known at Receiver-$2$, Receiver-$3$ and Receiver-$1$ respectively, we perform the usual MAC-type decoding on messages of User-$1$ and User-$2$ at Receiver-$1$, User-$2$ and User-$1$ at Receiver-$2$ and User-$3$ and User-$1$ at Receiver-$3$. They can be decoded with arbitrarily small probability of error if (\ref{eqn19})-(\ref{eqn23}) and
%%%%%%%%%%%%%%%%%%%%%%%%%%%%%%%%%%%%%%%%%%%%%

{\small
\begin{align}
\nonumber
&R_1 < \min \{I[X_1;Y_2|X_2,X_3,Q], I[X_1;Y_3|X_2,X_3,Q]\}\\
\label{redund4}
&R_2 < I[X_2;Y_1|X_3,X_1,Q]
\end{align}}are satisfied. The conditions in (\ref{redund3}) and (\ref{redund4}) are redundant because (\ref{eqn13})-(\ref{eqn18}) are satisfied $\forall ~p_{X_1}(x_1)p_{X_2}(x_2)p_{X_3}(x_3)$. The cardinality of $Q$ follows from direct application of Caratheodory Theorem \cite{CoT-book}.
\end{proof}
%%%%%%%%%%%%%%%%%%%%%%%%%%%%%%%%%%%%%%%%%%%%%%%%%%%%%%%%%%%%%%%%%%%%%%%%%%
\section{CAPACITY REGION OF 3-USER GIC WITH MIXED STRONG-VERY STRONG INTERFERENCE}
\label{sec4}
Consider a 2-user GIC with the following input-output equations 
\begin{equation}
\nonumber
Y_j=\sum_{i=1}^{2} h_{ij}X_i + N_j
\end{equation}where, $h_{ij}$ is the channel gain from User-$i$ to Receiver-$j$, $h_{ii}$=$1$, $h_{ij} \in \mathbb{C}$ $(j \neq i)$, $N_j$ $\sim$ ${\cal CN}(0,1)$ and $j \in \{1,2\}$. User-$i$ has a power constraint $P_i$ $(i \in \{1,2\})$.
We shall first state a lemma from \cite{GaY}, which was mentioned in the context of the 2-user GIC.
%%%%%%%%%%%%%%%%%%%%%%%%%%%%%%%
\begin{lemma}[\cite{GaY}]
\label{lemma2}
\begin{enumerate}
\item If $I[X_1;Y_1|X_2] \leq I[X_1;Y_2|X_2]$ when $X_i \sim $ ${\cal CN}(0,P_i)$ $(i \in \{1,2\})$, then $I[X_1;Y_1|X_2] \leq I[X_1;Y_2|X_2]~ \forall~ p_{X_1}(x_1)p_{X_2}(x_2)$. Similarly, when $I[X_2;Y_2|X_1] \leq I[X_2;Y_1|X_1]$ when $X_i \sim {\cal CN}(0,P_i)$ $(i \in \{1,2\})$, then $I[X_2;Y_2|X_1] \leq I[X_2;Y_1|X_1]~ \forall~ p_{X_1}(x_1)p_{X_2}(x_2)$.\\
\item If $I[X_1;Y_1|X_2] \leq I[X_1;Y_2]$ when $X_i \sim{\cal CN}(0,P_i)$ $(i \in \{1,2\})$, then $I[X_1;Y_1|X_2] \leq I[X_1;Y_2]~ \forall~ p_{X_1}(x_1)p_{X_2}(x_2)$. Similarly, when $I[X_2;Y_2|X_1] \leq I[X_2;Y_1]$ when $X_i \sim {\cal CN}(0,P_i)$ $(i \in \{1,2\})$, then $I[X_2;Y_2|X_1] \leq I[X_2;Y_1]~ \forall~ p_{X_1}(x_1)p_{X_2}(x_2)$.
\end{enumerate}
\end{lemma} 

Note that $I[X_1;Y_1|X_2]=I[X_1;X_1+N_1]$ and $I[X_1;Y_2]=I[X_1;X_2+h_{12}X_1+N_2]$.
%%%%%%%%%%%%%%%%%%%%%%%%%%%%%%%

Now, consider a 3-user Gaussian IC with the following input-output equations 
\begin{equation}
\nonumber
Y_j=\sum_{i=1}^{3} h_{ij}X_i + N_j
\end{equation}where, $h_{ij}$ is the channel gain from User-$i$ to Receiver-$j$, $h_{ii}$=$1$, $h_{ij} \in \mathbb{C}$ $(j \neq i)$, $N_j$ $\sim$ ${\cal CN}(0,1)$ and $j \in \{1,2,3\}$. User-$i$ has a power constraint $P_i$ $(i \in \{1,2,3\})$.

Let us consider the Case 1, where we assumed, without loss of generality, that User-$1$ causes a very strong interference at Receiver-$2$ and strong interference at Receiver-$3$, User-$2$ causes a very strong interference at Receiver-$3$ and strong interference at Receiver-$1$, and User-$3$ causes a very strong interference at Receiver-$1$ and strong interference at Receiver-$2$.
%%%%%%%%%%%%%%%%%%%%%%%%%%%%%%%
\begin{corollary}
\label{cor1}
 The capacity region of the 3-user GIC satisfying the conditions 

{\small
\begin{align}
\label{eqn24}
&|h_{12}|^2 \geq 1 + P_2 + |h_{32}|^2 P_3\\
\label{eqn25}
&|h_{23}|^2 \geq 1 + P_3 + |h_{13}|^2 P_1\\
\label{eqn26}
&|h_{31}|^2 \geq 1 + P_1 + |h_{21}|^2 P_2\\
\label{eqn27}
&|h_{13}| \geq 1\\
\label{eqn28}
&|h_{21}| \geq 1\\
\label{eqn29}
&|h_{32}| \geq 1
\end{align}}is given by
%%%%%%%%%%%

{\small
\begin{align}
\label{eqn30}
&R_1<I[X_{1G};Y_{1G}|X_{2G},X_{3G}]\\
\label{eqn31}
&R_2<I[X_{2G};Y_{2G}|X_{3G},X_{1G}]\\
\label{eqn32}
&R_3<I[X_{3G};Y_{3G}|X_{1G},X_{2G}]\\
\label{eqn33}
&R_1+R_2<I[X_{1G},X_{2G};Y_{1G}|X_{3G}]\\
\label{eqn34}
&R_2+R_3<I[X_{2G},X_{3G};Y_{2G}|X_{1G}]\\
\label{eqn35}
&R_3+R_1<I[X_{3G},X_{1G};Y_{3G}|X_{2G}]
\end{align}}where, $X_{1G} \sim {\cal CN}(0,P_1)$, $X_{2G} \sim {\cal CN}(0,P_2)$, $X_{3G} \sim {\cal CN}(0,P_3)$.
\end{corollary}
\begin{proof}
Condition (\ref{eqn24}) implies that condition (\ref{eqn1}) is satisfied for $X_i \sim {\cal CN}(0,P_i)$ $(i \in \{1,2,3\})$. This, in turn, implies that condition (\ref{eqn1}) is satisfied $\forall ~p_{X_1}(x_1)p_{X_2}(x_2)p_{X_3}(x_3)$, because we can treat $X_2+h_{32}X_3$ as a single channel input and apply the second part of Lemma \ref{lemma2} (note that $I[X_1;Y_1|X_2,X_3]=I[X_1;X_1+N_1]$ and $I[X_1;Y_2]=I[X_1;h_{12}X_1+X_2+h_{32}X_3+N_2]$). Similarly, the conditions (\ref{eqn25}) and (\ref{eqn26}) imply that the conditions (\ref{eqn2}) and (\ref{eqn3}) are respectively satisfied $\forall ~p_{X_1}(x_1)p_{X_2}(x_2)p_{X_3}(x_3)$. The conditions (\ref{eqn27})-(\ref{eqn29}) imply that (\ref{eqn4})-(\ref{eqn6}) are satisfied for $X_i \sim {\cal CN}(0,P_i)$ $(i \in \{1,2,3\})$. By application of the first part of Lemma \ref{lemma2}, the conditions (\ref{eqn27})-(\ref{eqn29}) imply that (\ref{eqn4})-(\ref{eqn6}) are satisfied $\forall ~p_{X_1}(x_1)p_{X_2}(x_2)p_{X_3}(x_3)$. Since Gaussian alphabets maximize the expressions in (\ref{eqn7})-(\ref{eqn12}), the capacity region is as given in (\ref{eqn30})-(\ref{eqn35}).
\end{proof}

Similarly, it can be easily shown that, for Case 2, when (\ref{eqn13})-(\ref{eqn18}) are satisfied for $X_i \sim {\cal CN}(0,P_i)$ $(i \in \{1,2,3\})$,i.e.,

{\small
\begin{align}
\label{eqn36}
&|h_{32}|^2 \geq 1 + P_2 + |h_{12}|^2 P_1\\
\label{eqn37}
&|h_{23}|^2 \geq 1 + P_3 + |h_{13}|^2 P_1\\
\label{eqn38}
&|h_{31}|^2 \geq 1 + P_1 + |h_{21}|^2 P_2\\
\label{eqn39}
&|h_{13}| \geq 1\\
\label{eqn40}
&|h_{21}| \geq 1\\
\label{eqn41}
&|h_{12}| \geq 1
\end{align}}
the capacity region of the channel is given by 

{\small
\begin{align}
\nonumber
&R_1<I[X_{1G};Y_{1G}|X_{2G},X_{3G}]\\
\nonumber
&R_2<I[X_{2G};Y_{2G}|X_{3G},X_{1G}]\\
\nonumber
&R_3<I[X_{3G};Y_{3G}|X_{1G},X_{2G}]\\
\nonumber
&R_1+R_2< \min \{I[X_{1G},X_{2G};Y_{1G}|X_{3G}],I[X_{1G},X_{2G};Y_{2G}|X_{3G}]\}\\
\nonumber
&R_3+R_1<I[X_{3G},X_{1G};Y_{3G}|X_{2G}]
\end{align}}
The above results for the GIC were independently proved in \cite{ChS}, but our proof is simpler. 

%%%%%%%%%%%%%%%%%%%%%%%%%%%%%%%%%%%%%%
\section{EXTENSION TO THE $K$-USER DMIC CASE}
\label{sec5}
In this section, we generalize capacity region of the $3$-user DMIC with mixed strong-very strong interference to the $K$-user scenario $(K \geq 3)$. For the $K$-user case, each receiver is constrained to receive one strong interference and $(K-2)$ very strong interferences. 

Let $\overline{X_{i}}=\{X_j|j \neq i\}$ $(i,j \in \{1,2,..K\})$. Let  $l_j$ ($j \in \{1,2,..K\}, l_j \in \{1,2,..K\}$) denote the strong-interferer at Receiver-$j$, i.e., one which satisfies the condition
%%%%%%%%%%%%%%%%%%%%%%%%%%%%%%%%%%%%%

{\small
\begin{align}
\label{eqnkuser1}
I[X_{l_j};Y_{l_j}|\overline{X_{l_j}}] \leq I[X_{l_j};Y_j|\overline{X_{l_j}}] ~\forall ~p_{X_1}(x_1)\cdot \cdot ~p_{X_K}(x_K)
\end{align}}
%%%%%%%%%%%%%%%%%%%%%%%%%%%%%%%%%%%%%
and, let $\overline{l_j}$ denote the set of all very strong interferers at Receiver-$j$, i.e., which includes each User-$m$, $m \in \{1,2,..K\}$, that satisfies the condition
%%%%%%%%%%%%%%%%%%%%%%%%%%%%%%%%%%%%%

{\small
\begin{align}
\label{eqnkuser2}
I[X_m;Y_j] \geq I[X_m;Y_m|\overline{X_m}] ~\forall ~p_{X_1}(x_1) \cdot \cdot ~p_{X_K}(x_K).
\end{align}}
Let $X_{\overline{l_j}}=\{X_i|i \in \overline{l_j}\}$.
%%%%%%%%%%%%%%%%%%%%%%%%%%%%%%%%%%%%%%%%%%%%%%%%%%%%%%%%%%%%%%%%%%%%%%%%%%%%%

\begin{theorem}
 \label{thm3}
The capacity region of the $K$-user DMIC $(K \geq 3)$, where, each receiver is constrained to receive one strong interference and $(K-2)$ very strong interferences (the conditions in (\ref{eqnkuser1}) and (\ref{eqnkuser2}) are satisfied at every Receiver-$j$) is given by

{\small
\begin{align}
\label{thm3_eqn1}
&R_i < I[X_i;Y_i|\overline{X_i},Q] ~\forall ~i \in \{1,2,..K\}\\
\label{thm3_eqn2}
&R_j+R_{l_j}<I[X_j,X_{l_j};Y_j|X_{\overline{l_j}},Q] ~\forall ~j \in \{1,2,..K\} 
\end{align}}for some distribution $p_Q(q)p_{X_1|Q} \cdot \cdot ~p_{X_K|Q}(x_K|q)$, where, $Q$ is the time-sharing random-variable with $|Q| \leq 2K+1$. 
\end{theorem}
\begin{proof}
The converse and achievability are given in the Appendix.
\end{proof}

The extension of the capacity region of $3$-user GIC with mixed strong-very strong interference to the $K$-user case where each receiver receives one strong interference and $(K-2)$ very strong interferences is straight-forward from (\ref{eqnkuser1}), (\ref{eqnkuser2}), Lemma \ref{lemma2} and Theorem \ref{thm3}.
%%%%%%%%%%%%%%%%%%%%%%%%%%%%%%%%%%%%%
\section{DISCUSSION}
\label{sec6}
The capacity region of a Gaussian channel when the inputs take values from finite complex constellations, with uniform distribution over the constellation, is called the Constellation Constrained (CC) capacity \cite{Big}. The CC capacity for the Gaussian-MAC (G-MAC) was analyzed in \cite{HaR}, for the 2-user GIC with strong interference in \cite{KnS} and \cite{AbR}. With finite constellations, suboptimality of Frequency Division Multiple Access (FDMA) scheme was shown for the G-MAC in \cite{HaR} and for the 2-user GIC with strong interference in \cite{AbR}. A similar analysis with finite constellations for the class of $K$-user GIC considered here is an interesting direction to pursue. An important direction of future research is to design practical schemes that would take us close to the CC capacity. This problem has been open even for the $2$-user GIC with strong interference.

\section*{Acknowledgement}
The authors wish to thank Dr. Rajesh Sundaresan and T. Damodaram Bavirisetti for the useful discussions. This work was supported  partly by the DRDO-IISc program on Advanced Research in Mathematical Engineering through a research grant as well as the INAE Chair Professorship grant to B.~S.~Rajan.

\begin{appendix}
\label{append1}
\vspace{-0.3cm}
\begin{equation}
\nonumber
\mbox{\textbf{Proof of Theorem \ref{thm3}}}
\end{equation}
\textit{Converse:} Equation (\ref{thm3_eqn1}) is the usual outer bound. Now, at each Receiver-$j$, the following steps hold good.
%%%%%%%%%%%

{\small
\begin{align}
\nonumber
n(R_j+R_{l_j}) &\stackrel{(a)}{<} I[X_j^n;Y_j^n]+I[X_{l_j}^n;Y_{l_j}^n]+n \epsilon_n
\end{align}
\begin{align}
\nonumber
&<I[X_j^n;Y_j^n,X_{\overline{l_j}}^n]+I[X_{l_j}^n;Y_{l_j}^n,X_j^n,X_{\overline{l_j}}^n]+n \epsilon_n\\
\nonumber
&\stackrel{(b)}{=}I[X_j^n;Y_j^n|X_{\overline{l_j}}^n]+I[X_{l_j}^n;Y_{l_j}^n|X_j^n,X_{\overline{l_j}}^n]+n \epsilon_n\\
\nonumber
&\stackrel{(c)}{\leq}I[X_j^n;Y_j^n|X_{\overline{l_j}}^n]+I[X_{l_j}^n;Y_j^n|X_j^n,X_{\overline{l_j}}^n]+n \epsilon_n\\
\nonumber
&=I[X_j^n,X_{l_j}^n;Y_j^n|X_{\overline{l_j}}^n]+n \epsilon_n\\
\nonumber
&\stackrel{(d)}{=} h(Y_j^n|X_{\overline{l_j}}^n)-\sum_{i=1}^{n} h(Y_{ji}|X_{1i},X_{2i} \cdot \cdot ~X_{Ki})\\
\nonumber
&~~~~~~~~~~~~~~~~~~~~~~~~~~~~~~~~~~~~~~~~~~~~~~~~~+n \epsilon_n\\
\nonumber
&\hspace{-0.2cm}\stackrel{(e)}{<} \sum_{i=1}^{n}\left( h(Y_{ji}|X_{\overline{l_j}i})- h(Y_{ji}|X_{1i},X_{2i}, \cdot \cdot ~X_{Ki}) \right)\\
\nonumber
&~~~~~~~~~~~~~~~~~~~~~~~~~~~~~~~~~~~~~~~~~~~~~~~~~+n \epsilon_n\\
\nonumber
&=\sum_{i=1}^{n}I[X_{ji},X_{l_ji};Y_{ji}|X_{\overline{l_j}i}]+n \epsilon_n\\
\nonumber
&=n\frac{1}{n}\sum_{i=1}^{n}I[X_{ji},X_{l_ji};Y_{ji}|X_{\overline{l_j}i}]+n \epsilon_n\\
\nonumber
\Rightarrow R_j+R_{l_j}&< I[X_j,X_{l_j};Y_j|X_{\overline{l_j}},Q]+ \epsilon_n
\end{align}}where, $(a)$ follows from Fano's inequality, $(b)$ follows from the fact that $X_j^n$, $X_{l_j}^n$ and $X_{\overline{l_j}}^n$ are independent, $(c)$ follows from (\ref{eqnkuser1}), Lemma \ref{lemma1} and the fact that $(X_j^n,X_{\overline{l_j}}^n)=\overline{X_{l_j}}^n$, $(d)$ follows from the memoryless property of the channel and $(e)$ follows from the fact that removing conditioning increases entropy. Finally, taking the limit as $n \rightarrow \infty$, $P_e^{(n)} \rightarrow 0$, we have 

{\small
\begin{align}
\nonumber
R_j+R_{l_j}&< I[X_j,X_{l_j};Y_j|X_{\overline{l_j}},Q].
\end{align}}
\textit{Achievability:}  Fix $p_{Q}(q)p_{X_1|Q}(x_1|q) \cdot \cdot ~p_{X_K|Q}(x_K |q)$. At User-$i$, $i \in \{1,2,..K\}$, generate $2^{nR_i}$ independent codewords $X_i^n (W_{ip})$, $p \in \{1,2,3,..2^{nR_i}\}$, of length $n$, generating each element i.i.d. $\sim \prod_{k=1}^{n} p_{X_i|Q}(x_{ik}|q)$, where, $W_{ip}$ indicates message-$p$ at User-$i$. Let $l_j=\{l_j(1),l_j(2),..l_j(K-2)\}$, where, $l_j(m)$ denotes the $m^{\mbox{th}}$ very strong interferer. At Receiver-$j$, decode $l_j(1)$, $l_j(2)$, $\cdot \cdot$, $l_j(K-2)$ in succession using the usual weak-typical set decoding. They can be decoded with arbitrarily small probability of error if

{\small
\begin{align}
\nonumber
&R_{l_j(1)}<I[X_{l_j(1)};Y_j|Q]\\
\nonumber
&R_{l_j(2)}<I[X_{l_j(2)};Y_j|X_{l_j(1)},Q]\\
\label{redundthm3_1}
&~~~\vdots\\
\nonumber
&R_{l_j(K-2)}<I[X_{l_j(K-2)};Y_j|X_{l_j(1)},X_{l_j(2)}, \cdot \cdot ~X_{l_j(K-3)},Q]
\end{align}}
where, $I[X_{l_j(m)};Y_j|X_{l_j(1)},X_{l_j(2)}, \cdots ~X_{l_j(m-1)},Q]$ $>$ $I[X_{l_j(m)};Y_j|Q]$ $\geq$ $I[X_{l_j(m)};Y_{l_j(m)}|\overline{X_{l_j(m)}},Q]$ (from (\ref{eqnkuser2})), for every $m$ $\in$ $\{1,2,..~K-2\}$. Now, given that the messages of the very strong interferers are known at Receiver-$j$, we perform the usual MAC-type decoding on messages of User-$j$ and User-$l_j$ at Receiver-$j$. They can be decoded with arbitrarily small probability of error if (\ref{thm3_eqn1}), (\ref{thm3_eqn2}) and
%%%%%%%%%%%%%%%%%%%%%%%%%%%%%%%%%%%%%%%%%%%%%

{\small
\begin{align}
\label{redundthm3_2}
&R_{l_j} < I[X_{l_j};Y_j|\overline{X_{l_j}},Q]
\end{align}}are satisfied. The conditions in (\ref{redundthm3_1}) and (\ref{redundthm3_2}) are redundant because (\ref{eqnkuser1}) and (\ref{eqnkuser2}) are satisfied $\forall ~p_{X_1}(x_1)p_{X_2}(x_2) \cdots ~p_{X_K}(x_K)$. The cardinality of $Q$ follows from direct application of Caratheodory Theorem \cite{CoT-book}.
%%%%%%%%%%%%%%%%%%%%%%%%%%%%%%%%%%%%%%%%%%%%%%%%%%%%%%%%%%%%%%%%%%%%%%%%%%
\end{appendix}
\end{document}